\documentclass[a4paper,conference]{IEEEtran}

\usepackage{amsmath}
\usepackage{amssymb}
\usepackage{amsthm, graphicx}
\usepackage{algorithm,algorithmic}

\newtheorem{prop}{Proposition}
\newtheorem{lem}[prop]{Lemma}
\newtheorem{thm}[prop]{Theorem}

\theoremstyle{definition}
\newtheorem{defi}[prop]{Definition}

\newtheorem{ex}[prop]{Example}
\newtheorem{rem}[prop]{Remark}

\newcommand{\F}{\mathbb{F}}

\newcommand{\rk}{\mathrm{rank}}

\newcommand{\rs}{\mathrm{rs}}

\newcommand{\qdeg}{\mathrm{qdeg}}

\newcommand{\Lp}{\mathcal{L}_q(x,q^m)}

\IEEEoverridecommandlockouts
\renewcommand{\baselinestretch}{0.897}

\title{Iterative List-Decoding of Gabidulin Codes via Gr\"obner Based Interpolation}
\author{\IEEEauthorblockN{Margreta Kuijper and Anna-Lena Trautmann\thanks{ALT is also with the Department of Electrical and Computer Systems Engineering, Monash University. She was supported by Swiss National Science Foundation Fellowship no. 147304.}} \IEEEauthorblockA{Department of Electrical and Electronic Engineering, University of Melbourne, Australia.}}

\begin{document}

\maketitle

\begin{abstract}
We show how Gabidulin codes can be list decoded by using an iterative parametrization approach. 
For a given received word, our decoding algorithm processes its entries one by one, constructing four polynomials at each step. This then yields a parametrization of interpolating solutions for the data so far. From the final result a list of all codewords that are closest to the received word with respect to the rank metric is obtained.
\end{abstract}

\section{Introduction}

Over the last decade there has been increased interest in Gabidulin codes, mainly because of their relevance to network coding~\cite{ko08,si08j}. Gabidulin codes are optimal rank-metric nonbinary codes over a field $\F_q^m$ (where $q$ is a prime power). They were first derived by Gabidulin in \cite{ga85a} and independently by Delsarte in \cite{de78}. 
These codes can be seen as the $q$-analog of Reed-Solomon codes, using $q$-linearized polynomials instead of arbitrary polynomials. They are optimal in the sense that they are not only MDS codes with respect to the Hamming metric, but also achieve the Singleton bound with respect to the rank metric and are thus MRD codes. They are not only of interest in network coding but also in space-time coding \cite{lu03}, crisscoss error correction \cite{ro91} and distributed storage \cite{si12}. 

The decoding of Gabidulin codes has obtained a fair amount of attention in the literature, starting with work on decoding inside the unique decoding radius in~\cite{ga85a,ga92} and more recently~\cite{lo06,ri04p,si10,si11,si09p}.  Decoding beyond the unique decoding radius was investigated in e.g.\ \cite{lo06p,ko08, ma12,wa13p,xi11}. Related work on list-decoding of lifted Gabidulin codes can be found in \cite{tr13p}.

Using the close resemblance between Reed-Solomon codes and Gabidulin codes, the paper~\cite{lo06} translates Gabidulin decoding into a set of polynomial interpolation conditions. Essentially, this setup is also used in the papers~\cite{ko08,xi11} that present iterative algorithms that perform Gabidulin list decoding with a list size of 1. In this paper we present an iterative algorithm that bears similarity to the ones in~\cite{lo06,ko08,xi11} but yields {\em all} closest codewords rather than just one. 
The latter is due to our parametrization approach.

The paper is structured as follows. In the next section we present several preliminaries on $q$-linearized polynomials, Gabidulin codes, the rank metric and we recall the polynomial interpolation conditions from~\cite{lo06}. We also detail the iterative construction of the $q$-annihilator polynomial and the $q$-Lagrange polynomial. Section \ref{sec:prelim} closes with several preliminaries on Gr\"obner bases. In Section \ref{sec:interpolation} we reformulate the Gabidulin list decoding requirements in terms of a module represented by four $q$-linearized polynomials. 
In Section \ref{sec:algorithm} we present the algorithm and our main result which details how the algorithm yields a list of all closest message polynomials. We conclude this paper in Section \ref{sec:conclusion}.


\section{Preliminaries}\label{sec:prelim}

\subsection{$q$-linearized polynomials}

Let $q$ be a prime power and let $\F_q$ denote the finite field with $q$ elements. It is well-known that there always exists a primitive element $\alpha$ of the extension field $\F_{q^m}$, such that $\F_{q^m}\cong \F_q[\alpha] $. Moreover, $\F_{q^m}$  is isomorphic (as a vector space) to the vector space $\F_q^m$. 
One then easily gets the isomorphic description of matrices over the base field $\F_q$ as vectors over the extension field, i.e.\ $\F_q^{m\times n}\cong \F_{q^m}^n$. Since we will work with matrices over different underlying fields we denote the rank of a matrix $X$ over $\F_q$ by $\rk_q(X)$.

For some vector $(v_1,\dots, v_n) \in \F_{q^m}^n$ we denote the $k \times n$ \emph{Moore matrix} by
\[M_k(v_1,\dots, v_n) := \left( \begin{array}{cccc}  v_1 & v_2 &\dots &v_n \\ v_1^{[1]} & v_2^{[1]} &\dots &v_n^{[1]} \\ \vdots \\  v_1^{[k-1]} & v_2^{[k-1]} &\dots &v_n^{[k-1]} \end{array}\right)   ,\]
where $[i]:= q^i$. A \emph{$q$-linearized polynomial} over $\F_{q^m}$ is defined to be of the form
\[f(x) = \sum_{i=0}^{n} a_i x^{[i]}   \quad, \quad a_i \in\F_{q^m} , \]
where $n$ is called the \emph{$q$-degree} of $f(x)$, assuming that $a_n\neq 0$, denoted by $\qdeg (f)$. This class of polynomials was first studied by Ore in \cite{or33}. 
One can easily check that $f(x_1 + x_2)= f(x_1)+f(x_2)$ and $f(\lambda x_1) = \lambda f(x_1)$ for any $x_1,x_2 \in \F_{q^m}$ and $\lambda \in \F_q$, hence the name \emph{linearized}. The set of all $q$-linearized polynomials over $\F_{q^m}$ is denoted by $\Lp$. This set forms a non-commutative ring with the normal addition $+$ and composition $\circ$ of polynomials. 
Because of the non-commutativity, products and quotients of elements of $\Lp$ have to be specified as being ``left" or `right" products or quotients. To not be mistaken with the standard division, we call the inverse of the composition \emph{symbolic division}. I.e.\ $f(x)$ is symbolically divisible by $g(x)$ with right quotient $m(x)$ if $$ g(x) \circ m(x) = g(m(x)) = f(x).$$
Efficient algorithms for all these operations (left and right symbolic multiplication and division) exist and can be found e.g.\ in \cite{ko08}.

\begin{lem}[cf.\ \cite{li97b} Thm. 3.50]\label{lem:rootspace}
Let $f(x) \in \Lp$ and $\F_{q^s}$ be the smallest extension field of $\F_{q^m}$ that contains all roots of $f(x)$. Then the set of all roots of $f(x)$ forms a $\F_q$-linear vector space in $\F_{q^s}$.
\end{lem}

\begin{lem}[\cite{li97b} Thm. 3.52]\label{lem:nullpoly}
Let $U$ be a $\F_q$-linear subspace of $\F_{q^m}$. Then 
\( \prod_{g \in U} (x-g)\)
is an element of $\Lp$.
\end{lem}

Note that, if $g_1,\dots,g_n$ is a basis of $U$, one can rewrite 
$$ \prod_{g \in U} (x-g) = \lambda \det(M_{t+1}(g_1,\dots,g_n,x))$$ 
for some constant $\lambda\in\F_{q^m}$. We call this polynomial the \emph{$q$-annihilator polynomial of $U$}, denoted by $\Pi_{(g_1, g_2, \ldots , g_n)}(x)$. Clearly its $q$-degree equals $n$.

%
%
%
%

We also have a notion of $q$-Lagrange polynomial:
%
Let $\mathbf g=(g_1,\dots,g_n)$ and $\mathbf r=(r_1,\dots,r_n)$. Define the matrix $\mathfrak{D}_i(\mathbf g, x)$ as $  M_{n}(g_1,\dots,g_n,x)$ without the $i$-th column.
 We define the \emph{$q$-Lagrange polynomial} as  
\[\Lambda_{\mathbf g, \mathbf r}(x) := \sum_{i=1}^n (-1)^{n-i}  r_i \frac{\det(\mathfrak{D}_i(\mathbf g, x))}{\det (M_n(\mathbf g))} \quad \in \F_{q^m}[x] .\]
%
It can be easily verified that the above polynomial is $q$-linearized and that $\Lambda_{\mathbf g, \mathbf{r}}(g_i) = r_i $ for $i=1,\dots,n$.

Note that, although not under the same name, the previous two polynomials were also defined in e.g.\ \cite{wa13phd}.


In the following we will use matrix composition, which is defined analogously to matrix multiplication:
$$\left[\begin{array}{cc}  a(x) & b(x) \\ c(x) & d(x) \end{array}\right] \circ\left[\begin{array}{cc}  e(x) & f(x) \\ g(x) & h(x) \end{array}\right]:= $$ $$\left[\begin{array}{cc}  a(e(x)) + b(g(x)) & a(f(x))+ b(h(x)) \\ c(e(x))+ d(g(x)) & c(f(x))+ d(h(x)) \end{array}\right] .$$ 

We can recursively construct the $q$-annihilator and the $q$-Lagrange polynomial as follows.

\begin{prop}\label{prop:Lagrec}
Let $g_1,\dots,g_n\in\F_{q^m}$ be linearly independent and $r_1,\dots,r_n\in\F_{q^m}$. Define
$$\Pi_1(x):=  x^q-g_1^{q-1} x \quad , \quad\Lambda_{ 1}(x):=\frac{r_1}{g_1} x,$$
and for $i=1, \ldots , n-1$
$$ \left[ \begin{array}{cc} \Pi_{i+1}(x) \\ \Lambda_{i+1}(x)  \end{array}\right] :=  \left[ \begin{array}{cc} x^q- \Pi_i(g_{i+1})^{q-1}x & 0 \\ -  \frac{ \Lambda_{ {i}}(g_{i+1}) - r_{i+1}}{\Pi_{i} (g_{i+1})} x  & x  \end{array}\right] \circ  \left[ \begin{array}{cc} \Pi_{i}(x)  \\ \Lambda_{i}(x)   \end{array}\right]  .$$
Then for $i=1, \ldots , n$ we have $\Pi_i(x) = \Pi_{(g_1, g_2, \ldots , g_i)}(x)$ and $\Lambda_i(x) = \Lambda_{(g_1, g_2, \ldots , g_i),(r_1,\dots, r_{i})}(x)$.
\end{prop}
\begin{proof}
We prove this by induction on $i$. The theorem clearly holds for $i=1$. Suppose that the theorem holds for a value of $i$ with $1\leq i < n$. By definition $\Pi_{i+1}(x)= \Pi_i (x)^q - \Pi_i (g_{i+1})^{q-1} \Pi_i (x)$, so that (using the induction hypothesis)  $\Pi_{i+1}(x)$ is a monic $q$-linearized polynomial of $q$-degree $i+1$ such that for $1\leq j \leq i+1$ we have $\Pi_{i+1}(g_j)=0$. It follows that then $\Pi_{i+1}(x)$ must coincide with $\Pi_{(g_1, g_2, \ldots , g_{i+1})}(x)$. 

We next show that the formula for $\Lambda_{i+1}(x)$ yields the $q$-Lagrange polynomial at level $i+1$. 
%
Assume that $\Lambda_i(x)$ is the $q$-Lagrange polynomial at level $i$ and look at $\Lambda_{ {i+1}}(x)$, which is $q$-linearized since $\Lambda_i(x)$ and $ \Pi_{i}(x)$ are $q$-linearized. As $\qdeg( \Pi_{i}(x))= i > \qdeg(\Lambda_i(x))$ it holds that $\qdeg(\Lambda_{i+1}(x)) = i$. Furthermore, because $\Pi_{i}(g_j)=0$ for $j=1,\dots,i$ and $\Lambda_{ {i}}(g_j)=r_j$ for $j=1,\dots,i$ , 
$$\Lambda_{ {i+1}}(g_j) = \Lambda_{ {i}}(g_j) = r_j , \quad \textnormal{ and }$$
$$\Lambda_{ {i+1}}(g_{i+1}) = \Lambda_{ {i}}(g_{i+1}) -  \frac{ \Lambda_{ {i}}(g_{i+1}) - r_{i+1}}{\Pi_{i} (g_{i+1})}   \Pi_{i}(g_{i+1}) = r_{i+1} .$$
Therefore, $\Lambda_{ {i+1}}(x)$ evaluates in the same values as $\Lambda_{(g_1,\dots,g_{i+1}), (r_1,\dots, r_{i+1})}(x)$ 
for $g_1,\dots,g_{i+1}$. Because of the linearity of both these polynomials they evaluate in the same values for all elements of  $\langle g_1,\dots,g_{i+1}\rangle$, and as the $g_i$ are linearly independent, these are $q^{i+1}$ many values. Since the degree of both polynomials is $q^i < q^{i+1}$, it follows that they must be the same polynomial.
\end{proof}

\subsection{Gabidulin codes}

Let $g_1,\dots, g_n \in \F_{q^m}$ be linearly independent over $\F_q$. We define a \emph{Gabidulin code} $C\subseteq \F_{q^m}^{n}$ as the linear block code with generator matrix $M_k(g_1,\dots, g_n)$.        
Using the isomorphic matrix representation we can interpret $C$ as a matrix code in $\F_q^{m\times n}$.The \emph{rank distance} $d_R$ on  $\F_q^{m\times n}$ is defined by
\[d_R(X,Y):= \rk_q(X-Y) \quad, \quad X,Y \in \F_q^{m\times n} \]
and analogously for the isomorphic extension field representation. 
It holds that the code $C$ constructed before has dimension $k$ over $\F_{q^m}$ and minimum rank distance (over $\F_q$) $n-k+1$. One can easily see by the shape of the parity check and the generator matrices that an equivalent definition of the code is
\[C =  \{(f(g_1),\dots,f(g_n))\in \F_{q^m}^n \mid f(x) \in \Lp_{<k}  \} ,\]
where $\Lp_{<k} := \{f(x) \in \Lp, \qdeg(f(x)) < k\}$. 
For more information on bounds and constructions of rank-metric codes the interested reader is referred to \cite{ga85a}.

Consider a received word $\mathbf r = (r_1,\dots,r_n) \in \F_{q^m}^n$ as the sum $\mathbf r = \mathbf c + \mathbf e$, where $\mathbf c = (c_1,\dots,c_n)\in C$ is a codeword and $\mathbf e = (e_1,\dots,e_n)\in \F_{q^m}^n$ is the error vector. 
We now recall the polynomial interpolation setup from~\cite{lo06} via a more general formulation in the next theorem.

\begin{thm}[\cite{ku14,lo06}]\label{thm2}
Let $f(x)\in \Lp, \qdeg(f(x))< k$ and $c_i=f(g_i)$ for $i=1,\dots,n$.
Then $d_R(\mathbf c, \mathbf r) = t$ if and only if there exists a $D(x) \in \Lp$, such that $ \qdeg(D(x))= t$ and
\[D(r_i) = D(f(g_i)) \quad \forall i\in\{1,\dots,n\}.\]
Furthermore, this $D(x)$ is unique.
\end{thm}
%

\begin{rem}
The previous theorem states that the roots of $D(x)$ form a vector space of degree $t$ which is equal to the span of $e_1,\dots,e_n$ (for this note that $e_i=f(g_i)-r_i$). This is why $D(x)$ is also called the \emph{error span polynomial} (cf.\ e.g.\ \cite{si09}). The analogy in the classical Hamming metric set-up is the \emph{error locator polynomial}, whose roots indicate the locations of the errors, and whose degree equals the number of errors.
\end{rem}

%

\subsection{Gr\"obner bases}

We will now recall some definitions and results on Gr\"obner bases of $\Lp^2$, since we will need them 
later on in this paper.
Elements of $\Lp^2$ are of the form 
$$[f(x) \;\; g(x)] = f(x) e_1 + g(x) e_2$$ 
where $f(x)=\sum f_i x^i,g(x) = \sum g_i x^i \in \Lp$ and $e_1,e_2$ are the two unit vectors of length $2$. 

\begin{defi}
The \emph{$(k_1,k_2)$-weighted $q$-degree} of  $[f(x) \;\; g(x)]$ is defined as  $\max\{k_1+ \qdeg(f(x)) , k_2+ \qdeg(g(x))\}$.
\end{defi}

The monomials of $[f(x) \;\; g(x)]$ are of the form $x^{[i]} e_1$ and $x^{[j]} e_2$ for all $i$ such that $f_i\neq 0$ and $j$ such that $g_j\neq 0$, respectively.

\begin{defi}
The \emph{term-over-position (TOP) monomial order} is defined as 
$$x^{[i_1]} e_{j_1}< x^{[i_2]} e_{j_2} :\iff i_1<i_2 \textnormal{ or }  [i_1 = i_2 \textnormal{ and } j_1<j_2 ]  .$$ 
The \emph{$(k_1,k_2)$-weighted TOP monomial order} is defined as 
$$x^{[i_1]} e_{j_1}<_{(k_1,k_2)} x^{[i_2]} e_{j_2} :\iff $$ $$i_1+k_{j_1}<i_2 + k_{j_2} \textnormal{ or }  [i_1+k_{j_1} = i_2+k_{j_2} \textnormal{ and } j_1<j_2 ]  .$$ 
\end{defi}

We can order all monomials of an element $V\in\Lp^2$ in decreasing order with respect to the (weighted or non-weighted) TOP monomial order. Rename them such that $x^{[i_1]}e_{j_1}> x^{[i_2]}e_{j_2}> \dots $. Then
\begin{enumerate}
\item the \emph{leading monomial} $\mathrm{lm}(V)=x^{[i_1]}$ is the greatest monomial of $V$.
\item the \emph{leading position} $\mathrm{lpos}(V)={j_1}$ is the vector coordinate of the leading monomial.
\end{enumerate}

\begin{defi}\label{defi10}
Let $M\subseteq \Lp^2$ be a left module.
A subset $B\subset M$ is called a \emph{Gr\"obner basis} of $M$ if the leading monomials of $B$ span a left module that contains all leading monomials of $M$.
\end{defi}

 In the context of this paper we are only interested in modules with a basis consisting of two vectors, say $b_1,b_2\in\Lp^2$. It can be easily seen from Definition \ref{defi10} that such a basis $\{b_1,b_2\}$ is a Gr\"obner basis if and only if $\mathrm{ lpos} (b_1)\neq \mathrm{ lpos}(b_2)$. In fact, for this restricted special case such a basis coincides with a {\em minimal} Gr\"obner basis (see e.g.\ \cite{ku11}).


\section{Iterative Construction of the Interpolation Module}\label{sec:interpolation}



For the remainder of the paper 
let $g_1,\dots, g_n \in\F_{q^m}$ be linearly independent over $\F_q$ and let $M_k(g_1,\dots,g_n)$ be the generator matrix of the Gabidulin code $C\subseteq \F_{q^m}^n$. Let  $\mathbf{r}=(r_1,\dots,r_n) \in \F_{q^m}^n$ be the received word and denote $\mathbf{g}=(g_1,\dots,g_n) $.
%
Furthermore we need the following fact.

\begin{lem}[\cite{ku14}]\label{lem3}
Let  $L(x) \in \Lp$, such that $L(g_i)=0$ for all $i$. Then
\[\exists H(x)\in \Lp : L(x) = H(x)\circ \prod_{g \in \langle g_1,\dots,g_n\rangle}(x-g)  . \]
\end{lem}

In the following we abbreviate the row span of a (polynomial) matrix $A$ by $\rs(A)$.

\begin{defi}
The \emph{interpolation module} $\mathfrak{M}(\bf r)$ for $\mathbf r$ is defined as the left submodule of $\Lp$, given by
\[
\mathfrak{M}(\bf r) := \rs \left[\begin{array}{cc}  \Pi_\mathbf{g} (x) & 0 \\ -\Lambda_{\bf g,r}(x) & x \end{array}\right].
\]
\end{defi}

We identify any $[f(x) \quad g(x)] \in  \mathfrak{M}(\bf r)$ with the bivariate linearized $q$-polynomial $Q(x,y)= f(x) + g(y)$.
It was shown in our recent paper~\cite{ku14} that the name interpolation module is justified for $\mathfrak{M}(\bf r)$:

\begin{thm}[\cite{ku14}]\label{thm5}
$\mathfrak{M}(\bf r)$ consists exactly of all $Q(x,y)= f(x) +g(y)$ with $f(x), g(x) \in \Lp$, such that $Q(g_i,r_i)=0$ for $i=1,\dots,n$.
\end{thm}

%


The following statements can also be found in our recent paper~\cite{ku14}:

\begin{thm}\label{thm:main}
The elements $[N(x) \quad -D(x)]$ of $\mathfrak{M}(\bf r)$  that fulfill
\begin{enumerate}
\item $\qdeg(N(x))\leq t+k-1$,
\item $\qdeg(D(x))=t$,
\item $N(x)$ is symbolically divisible on the right by $D(x)$, i.e.\ there exists $f(x)\in\Lp$ such that $D(f(x))=N(x)$,
\end{enumerate}
are in one-to-one correspondence with the codewords of rank distance $t$ to $\mathbf r$. 
\end{thm}
%


Therefore, list decoding within rank radius $t$ is equivalent to finding all elements  $[N(x) \quad -D(x)]$ in $\mathfrak{M}(\bf r)$ with $(0,k-1)$-weighted $q$-degree less than or equal to $t+k-1$ and $\qdeg(N(x))\leq \qdeg(D(x))+k-1$, such that $N(x)$ is symbolically divisible on the right by $D(x)$.
It follows that, to find all closest codewords to a given $\mathbf r \in \F_{q^m}^n$, we need to find all elements  $[N(x) \quad -D(x)] \in \mathfrak{M}(\bf r)$ of minimal $(0,k-1)$-weighted $q$-degree such that $\qdeg(N(x))\leq \qdeg(D(x))+k-1$ and $N(x)$ is symbolically divisible on the right by $D(x)$. This minimality requirement leads us to construct a minimal Gr\"obner basis for  $\mathfrak{M}(\bf r)$.
Note that this is a generalization of the interpolation-based decoding method from \cite{lo06}. The difference is that our method can also decode beyond the unique decoding radius.

In contrast to our previous paper \cite{ku14} the algorithm below is iterative in the sense that it adds an extra pair of interpolation points $(g_i,r_i)$ at the $i$-th step of the algorithm. More specifically, in the $i$-th step a minimal Gr\"obner basis is constructed for the interpolation module corresponding to $(g_1,\dots,g_i), (r_1,\dots,r_i)$. 

\begin{thm}\label{thm:it_min_bas}
For $i=1,\dots,n$ denote by $\mathfrak{M}_i$ the interpolation module for $(g_1,\dots,g_i)$ and $(r_1,\dots,r_i)$. Let 
\[\left[\begin{array}{cc}  P(x) & -K(x) \\ N(x) & -D(x) \end{array}\right]\]
be a basis for $\mathfrak{M}_{i-1}$ and
$$\Delta_i := N(g_i)- D(r_i) \quad, \quad \Gamma_i := P(g_i)-K(r_i) .$$ 
If $\Gamma_i\neq 0$, then the row vectors of 
\[\left[\begin{array}{cc}  x^q-  \Gamma_i^{q-1} x & 0 \\ \Delta_i x& -\Gamma_i x\end{array}\right] \circ \left[\begin{array}{cc}  P(x) & -K(x) \\ N(x) & -D(x) \end{array}\right]\]
form a basis of $\mathfrak{M}_i$.
If $\Delta_i\neq 0$, then the row vectors of 
\[\left[\begin{array}{cc}   \Delta_i x& -\Gamma_i x \\ 0& x^q-  \Delta_i^{q-1} x  \end{array}\right] \circ \left[\begin{array}{cc}  P(x) & -K(x) \\ N(x) & -D(x) \end{array}\right]\]
form a basis of $\mathfrak{M}_i$.
\end{thm}
\begin{proof}
We first consider the first case and show that both row vectors are in $\mathfrak{M}_i$. From the assumptions it follows that $P(g_j)=K(r_j)$ and that $N(g_j)=D(r_j)$ for $1\leq j<i$. Moreover, the two entries of the first row are given by
$$(x^q- \Gamma_i^{q-1} x) \circ P(x) = P(x)^q-\Gamma_i^{q-1} P(x), $$
$$(x^q- \Gamma_i^{q-1} x) \circ K(x) = K(x)^q- \Gamma_i^{q-1} K(x) ,$$
thus
$P(g_j)^q- \Gamma_i^{q-1} P(g_j) - K(r_j)^q+ \Gamma_i^{q-1} K(r_j)=0$ for $1\leq j\leq i$.
For the second row we get
$$\Delta_i P(g_j) - \Gamma_i N(g_j) - \Delta_i K(r_j) + \Gamma_i D(r_j) = $$
$$\Delta_i (P(g_j) -K(r_j)) - \Gamma_i (N(g_j) -D(r_j)) = \Delta_i \Gamma_i - \Gamma_i \Delta_i= 0$$
for $1\leq j\leq i$. Thus, the two row vectors are elements of $\mathfrak{M}_i$. 

It remains to show that the two row vectors span the whole interpolation module (and not just a submodule of it). For this, we note that there exist $\bar a(x), \bar b(x) \in \Lp$ such that $\bar a(x)\circ [ P(x) \; -K(x) ] + \bar b(x) \circ [ N(x) \; -D(x) ] = [ \Pi_{i-1}(x) \; 0 ]$. Let $b(x) = -(x^q-\Pi_{i-1}(g_i)^{q-1} x)\circ \bar b(\frac{1}{\Gamma_i}x) \in \Lp$ and let $a(x)\in\Lp$ such that $a(x)\circ (x^q-\Gamma_i^{q-1}x) =(x^q-\Pi_{i-1}(g_i)^{q-1} x)\circ\left( \bar b(\frac{\Delta_i}{\Gamma_i}x)+\bar a(x)\right)$. Note that $a(x)$ is well-defined by Lemma \ref{lem3} since $\Gamma_i$ is a root of the right side of the previous equation. Denote the first and second row of the new basis by $b_1$ and $b_2$, respectively. Then $a(x)\circ b_1 + b(x) \circ b_2 = [ \Pi_{i}(x) \; 0 ]$, i.e.\  $[ \Pi_{i}(x) \; 0 ]$ is in the module spanned by the new basis. Analogously we can construct $a(x), b(x)\in\Lp$ such that $a(x)\circ b_1 + b(x) \circ b_2 = [ \Lambda_{i}(x) \; -x ]$. Hence, we have shown that the new basis spans the whole interpolation module.

For the second case note that
$$ \rs \left(\left[\begin{array}{cc}   \Delta_i x& -\Gamma_i x \\ 0& x^q-  \Delta_i^{q-1} x  \end{array}\right] \circ \left[\begin{array}{cc}  P(x) & -K(x) \\ N(x) & -D(x) \end{array}\right] \right)$$
$$= \rs \left(\left[\begin{array}{cc}    x^q-  \Delta_i^{q-1} x & 0 \\  \Gamma_i x & -\Delta_i x  \end{array}\right] \circ \left[\begin{array}{cc}  N(x) & -D(x)\\ P(x) & -K(x)  \end{array}\right] \right),$$
which corresponds to the first case after exchanging $P(x)$ with $N(x)$ and $K(x)$ with $D(x)$ (and vice versa).
\end{proof}

\begin{rem}
In the notation of Proposition~\ref{prop:Lagrec}, applying the previous theorem to $P(x)=\Pi_{i-1}(x), K(x)=0, N(x)=\Lambda_{i-1}(x)$ and $D(x)=-x$, leads to a computation that is identical up to a constant to the one in Proposition~\ref{prop:Lagrec} in which the $q$-annihilator polynomial and the $q$-Lagrange polynomial are iteratively constructed.
\end{rem}


\section{The Algorithm}\label{sec:algorithm}

Using Theorem \ref{thm:it_min_bas} as our main ingredient, we now set out to design an iterative algorithm that computes a minimal Gr\"obner basis for $\mathfrak{M}_i$ at each step $i$. 
We note that the calculation of the matrices $B_i$ in our algorithm coincides with the calculations in the interpolations algorithms of \cite{ko08,xi11}.
The complete decoding algorithm, stated in Algorithm \ref{alg1}, first computes a minimal   Gr\"obner basis for $\mathfrak{M}_n$ and then uses a parametrization to find all closest codewords to the received word.


%
%
%
%
\begin{algorithm}
\caption{Iterative minimal list decoding of Gabidulin codes.}
\label{alg1}
\begin{algorithmic}
\REQUIRE Linearly independent $g_1,\dots, g_n\in \F_{q^m}$, received $r_1,\dots,r_n \in \F_{q^m}$.
\STATE Initialize   \textbf{list}$:=[ \;]$, $j:=0$, $B_0 :=  \left[\begin{array}{cc}x& 0 \\ 0 & x \end{array}\right] $ .
\STATE We denote $B_i := \left[\begin{array}{cc}  P_i(x) & -K_i(x) \\ N_i(x) & -D_i(x) \end{array}\right] $.
\FOR{$i$ from $1$ to $n$}
\STATE  
$\Delta_i := N_{i-1}(g_i)- D_{i-1}(r_i) \quad, \quad \Gamma_i := P_{i-1}(g_i)-K_{i-1}(r_i) .$
\IF{[$\qdeg(P_{i-1}(x))\leq \qdeg(D_{i-1}(x))+k-1\;  \AND \; \Gamma_i\neq 0$] or  $\Delta_i =0$}
\vspace{0.4cm}
\STATE
$B_i := \left[\begin{array}{cc}  x^q-\Gamma_i^{q-1} x & 0 \\ \Delta_i x& -\Gamma_i x\end{array}\right] \circ B_{i-1}$
\ELSE
\vspace{0.4cm}
\STATE
$B_i:=\left[\begin{array}{cc} \Delta_i x& -\Gamma_i x \\  0 & x^q-\Delta_i^{q-1} x \end{array}\right] \circ B_{i-1}$
\ENDIF
\ENDFOR
\STATE Set $b_1(x):=$ first row of $B_n$, $b_2(x):=$ second row of $B_n$, $\ell_1:= \qdeg(b_1)$, $\ell_2:= \qdeg(b_2)$.
\WHILE{\textbf{list}$=[\;]$}
\FORALL{$a(x)\in \Lp, \qdeg(a(x))\leq \ell_2-\ell_1+j$}
\FORALL{monic $c(x) \in \Lp, \qdeg(b(x))=j$}
\STATE $f(x) := a(x)\circ b_1(x) + c(x)\circ b_2(x)$
\IF{$f^{(1)}(x)$ is symb. (right) divisible by $f^{(2)}(x)$} 
\STATE add the respective symb. quotient to \textbf{list}
\ENDIF
\ENDFOR
\ENDFOR
\STATE $j:=j+1$
\ENDWHILE
\RETURN \textbf{list}
%
%
\end{algorithmic}
\end{algorithm}
\begin{thm}
Algorithm \ref{alg1} yields a list of all message polynomials such that the corresponding codeword is closest to the received word.
\end{thm}
\begin{proof}
Denote by $M_1$ the matrix we multiply by on the left in the first IF statement and by $M_2$ the one in the ELSE statement of the algorithm. 
We know from Theorem \ref{thm:it_min_bas} that at each step $B_i$ is a basis for the interpolation module $\mathfrak M_i$ (swap the roles of $\Gamma_i$ and $\Delta_i$ where needed). We now show that it is a minimal Gr\"obner basis with respect to the $(0,k-1)$-weighted TOP monomial order via induction on $i$. Assume that at step $i$ the first row has leading position $1$ and the second row has leading position $2$, i.e.\ $\qdeg(P_i(x)) > \qdeg(K_i(x))+k-1$ and $\qdeg(N_i(x)) \leq \qdeg(D_i(x))+k-1$. Furthermore assume that $\qdeg(P_i(x))\geq \qdeg(N_i(x))$. If $\qdeg(P_i(x))\leq \qdeg(D_i(x))+k-1$ we multiply on the left by $M_1$. Hence, 
$$ \qdeg(P_{i+1}(x)) = \qdeg(P_i(x)) +1 ,$$
$$ \qdeg(K_{i+1}(x)) = \qdeg(K_i(x)) +1 ,$$ where the latter is less than $ \qdeg(P_{i}(x)) -k+2 =   \qdeg(P_{i+1}(x)) -k+1$ by the assumption. Thus, the leading position of the first row of $B_{i+1}$ is $1$. 
Moreover, 
$$ \qdeg(N_{i+1}(x)) \leq \max\{\qdeg(P_i(x)) , \qdeg(N_i(x))\}$$ $$=\qdeg(P_i(x)) \leq  \qdeg(D_i(x))+k-1$$
and, since the assumptions imply that $\qdeg(K_i(x)) < \qdeg(D_i(x))$,
{\small $$ \qdeg(D_{i+1}(x)) = \max\{\qdeg(K_i(x)) , \qdeg(D_i(x))\} = \qdeg(D_i(x)) .$$}
Thus the leading position of the second row is $2$. Moreover, $\qdeg(P_{i+1}(x))\geq \qdeg(N_{i+1}(x))$. Since the assumptions are true for $B_0$ the statement follows via induction. 

Analogously one can prove that multiplication with $M_2$ yields a basis of $\mathfrak M_i$ with different leading positions in the two rows. 
Thus, after $n$ steps, $B_n$ is a minimal Gr\"obner basis for the interpolation module $\mathfrak M(\bf r)$. Consequently,  $B_n$ has the so-called Predictable Leading Monomial Property, see~\cite{ku11} and~\cite{al11}. As a result of this property, the parametrization used for $a(x)$ and $c(x)$ in the second part of the algorithm will then yield all closest codewords. For the sake of brevity we omit the details.
\end{proof}
\begin{rem}
It can be verified that, due to the linear independence of $g_1, \ldots , g_k$, the first $k$ steps of the algorithm coincide up to a constant with the computation in Proposition~\ref{prop:Lagrec}. In other words, up to a constant, at step $k$ the algorithm has computed the $q$-annihilator polynomial and the $q$-Lagrange polynomial corresponding to the data so far.
\end{rem}
\begin{ex}\label{ex15}
Consider the Gabidulin code in $\F_{2^3}\cong \F_2[\alpha]$ (with $\alpha^3=\alpha +1$) with generator matrix 
\[G= \left( \begin{array}{ccc} 1 & \alpha & \alpha^2 \\ 1 & \alpha^2 & \alpha^4\end{array}\right) \]
(i.e.\ $g_1=1, g_2=\alpha, g_3= \alpha^2$ and $k=2$) and the received word $\mathbf{ r} =(\: \alpha^3 \;0 \; \alpha \:) $. We iteratively compute
$$B_1=\left[\begin{array}{cc}  x^2+x & 0 \\ (\alpha+1)x & x \end{array}\right] , $$
$$B_2=\left[\begin{array}{cc}  x^4+(\alpha^2 + \alpha+1)x^2+(\alpha^2 + \alpha)x & 0 \\ (\alpha^2+\alpha)x^2 +(\alpha^2 + \alpha+1)x & (\alpha^2+\alpha)x \end{array}\right] ,$$
$$B_3=\left[\begin{array}{cc}  \alpha^2 x^4+ \alpha^5 x & x\\ \alpha x^4 +\alpha^4 x^2 + x & \alpha x^2+\alpha^6 x \end{array}\right] .$$
$B_3$ is a minimal $(0,1)$-weighted Gr\"obner basis of the interpolation module. 
We get $\ell_1 =2$ and $ \ell_2= 2$, i.e.\ we want to use all $a(x)\in\mathcal{L}_2(x,2^3)$ with $2$-degree less than or equal to $0$ and all monic $c(x)\in\mathcal{L}_2(x,2^3)$  with $2$-degree equal to $0$. Thus, $a(x)= a_0 x$ for $a_0\in\F_{2^3}$ and $c(x)= x$. We get divisibility for $a_0\in \F_{2^3}\backslash \{ \alpha^6\}$. The corresponding message polynomials and codewords are
\[m_1(x) =  x^2 + \alpha x \quad, \quad c_1=(\: \alpha^3 \; 1 \; \alpha^3),\]
\[m_2(x) = \alpha^5 x^2 + \alpha^2 x \quad, \quad c_1=(\: \alpha^3  \; \alpha \; \alpha),\]
\[m_3(x) = \alpha^3 x^2 + \alpha^4 x \quad, \quad c_1=(\: \alpha^2 +1 \; 0 \; \alpha^2),\]
\[m_4(x) = \alpha^4 x^2  \quad,\quad c_3=(\: \alpha^2 +\alpha \; \alpha^2+1 \; \alpha),\]
\[m_5(x) = \alpha^6 x^2 + \alpha^6 x \quad, \quad c_1=(\: 0 \; \alpha^3 \; 1),\]
\[m_6(x) = \alpha^2 x^2 +  \alpha^3 x \quad,\quad c_2=(\: \alpha^5 \; 0 \; \alpha),\]
\[m_7(x) = \alpha x^2 +  x \quad,\quad c_2=(\: \alpha^3 \; 1 \; \alpha^3).\]
All these codewords are rank distance $1$ away from $\bf r$ (note that some of them are Hamming distance $1, 2$ or even $3$ away from $\bf r$).
\end{ex}

\section{Conclusion}\label{sec:conclusion}

In this paper we used a parametrization approach to the decoding of Gabidulin codes with respect to the rank metric. We presented a iterative algorithm with simple update steps that is similar to the ones found in the literature. Our main result is that we use this algorithm to compute a list of message polynomials that correspond to {\em all} codewords that are closest to a given received word. In our algorithm we construct, via a simple update matrix, a minimal Gr\"obner basis at each step. This setup allows for straightforward conclusions on minimality and parametrization due to the Predictable Leading Monomial Property, as in~\cite{al11} and~\cite{ku11}.




\bibliographystyle{plain}
\bibliography{margreta_anna-lena}

\end{document}